\documentclass[letterpaper,11pt]{article}

\usepackage{xfrac}
\usepackage{float}
\usepackage{multicol}
\usepackage{algorithmic}
\usepackage[ruled,vlined,commentsnumbered,titlenotnumbered]{algorithm2e}
\usepackage{verbatim}
\usepackage{latexsym}

\usepackage{amsmath,amssymb,amsthm,mathrsfs,amsfonts,dsfont}
\usepackage{epsfig,ifpdf,graphics}
\usepackage{enumerate}
\usepackage{booktabs}

\newcommand{\etal}{{\em et al.\ }}

\let\myPushQED=\pushQED
\let\myPopQED=\popQED
\newcommand{\myignore}[1]{}
\newenvironment{proof*}
  {\let\pushQED=\myignore\begin{proof}\let\pushQED=\myPushQED}
  {\def\popQED{}\end{proof}\let\popQED=\myPopQED}

\newenvironment{description*}%
  {\vspace{-1ex}\begin{description}%
    \setlength{\itemsep}{-0.5ex}%
    \setlength{\parsep}{0pt}}%
  {\end{description}}
\newenvironment{itemize*}%
  {\vspace{-1ex}\begin{itemize}%
    \setlength{\itemsep}{-0.5ex}%
    \setlength{\parsep}{0pt}}%
  {\end{itemize}}
\newenvironment{enumerate*}%
  {\vspace{-1ex}\begin{enumerate}%
    \setlength{\itemsep}{-0.5ex}%
    \setlength{\parsep}{0pt}}%
  {\end{enumerate}}

\newcommand{\ignore}[1]{}

{\makeatletter
 \gdef\xxxmark{%
   \expandafter\ifx\csname @mpargs\endcsname\relax 
     \expandafter\ifx\csname @captype\endcsname\relax 
       \marginpar{xxx}
     \else
       xxx 
     \fi
   \else
     xxx 
   \fi}
 \gdef\xxx{\@ifnextchar[\xxx@lab\xxx@nolab}
 \long\gdef\xxx@lab[#1]#2{{\bf [\xxxmark #2 ---{\sc #1}]}}
 \long\gdef\xxx@nolab#1{{\bf [\xxxmark #1]}}
}

\DeclareMathOperator*{\argmin}{arg\,min}

\newtheorem{theorem}{Theorem}
\newtheorem{lemma}[theorem]{Lemma}

\newtheorem{claim}[theorem]{Claim}

\newtheorem{property}[theorem]{Property}

\newcommand{\eps}{\varepsilon}

\let\phi=\varphi

\newcommand{\Ot}{\tilde{O}}

\newcommand\patrascu{P\v{a}tra\c{s}cu}

\title{New routing techniques and their applications}

\author{Liam Roditty\thanks{{\tt
liam.roditty@biu.ac.il}}\\Bar Ilan University \and Roei Tov\thanks{{\tt
roei81@gmail.com}}\\Bar Ilan University}

\begin{document}

\maketitle

\begin{abstract}
\setlength{\parindent}{0pt}
 \setlength{\parskip}{5pt plus 2pt}
\noindent%
Let $G=(V,E)$ be an undirected graph with $n$ vertices and $m$ edges. We present two new routing techniques. Roughly speaking, given a partition $\mathcal{U}=\{U_1, \ldots, U_q\}$ of $V$ into $q$ sets each of size $\Ot(n / q)$, our first technique routes a message between vertices of $U\in \mathcal{U}$ on a $(1 + \eps)$-stretch path with routing tables of size $\Ot(\frac{1}{\eps}(n/q) + q)$. Given a partition $\mathcal{W}=\{W_1, \ldots, W_q\}$ of a set $W\subseteq V$ into $q$ sets each of size $\Ot(|W| / q)$ and assuming that the sets of $\mathcal{U}$ satisfy a certain hitting set property with respect to vertex vicinities,  our second technique routes a message, for every $i\in \{1,\ldots,q\}$, from any source in $U_i$ to any destination in $W_i$ on a $(1 + \eps)$-stretch path with routing tables of size $\Ot(\frac{1}{\eps}(|W|/q) + q)$.
Using these techniques we obtain the following new routing schemes:
\begin{itemize}
\item A routing scheme for unweighted graphs that uses $\Ot(\frac{1}{\eps} n^{2/3})$ space at each vertex and $\Ot(1/\eps)$-bit headers, to route a message between any pair of vertices $u,v\in V$ on a $(2 + \eps,1)$-stretch path, i.e., a path of length at most $(2 + \eps)\cdot d+1$, where $d$ is the distance between $u$ and $v$.
This should be compared to the $(2,1)$-stretch and $\Ot(n^{5/3})$ space distance oracle of \patrascu\ and Roditty [FOCS'10 and SIAM J. Comput. 2014] and to the $(2,1)$-stretch routing scheme of Abraham and Gavoille [DISC'11] that uses $\Ot( n^{3/4})$ space at each vertex. It follows from \patrascu, Thorup and Roditty [FOCS'12] that a $2$-stretch routing scheme with $\Ot( m^{2/3})$ space at each vertex is optimal, assuming a hardness conjecture on set intersection holds.
\item A routing scheme for weighted graphs with normalized diameter $D$, that uses $\Ot(\frac{1}{\eps} n^{1/3}\log D)$ space at each vertex and $\Ot(\frac{1}{\eps} \log D)$-bit headers, to route a message between any pair of vertices on a $(5+\eps)$-stretch path. This should be compared to the $5$-stretch and $\Ot(n^{4/3})$ space distance oracle of Thorup and Zwick [STOC'01 and J. ACM. 2005] and to the $7$-stretch routing scheme of Thorup and Zwick [SPAA'01] that uses $\Ot( n^{1/3})$ space at each vertex. Since a $5$-stretch routing scheme must use tables of $\Omega( n^{1/3})$ space  our result is almost tight.

%
\item For an integer $\ell>1$, a routing scheme for unweighted graphs that uses $\Ot(\ell \frac{1}{\eps} n^{\ell/(2\ell \pm 1)})$ space at each vertex and $\Ot(\frac{1}{\eps})$-bit headers, to route a message between any pair of vertices on a $(3 \pm 2 / \ell + \eps,2)$-stretch path.
\item A routing scheme for weighted graphs, that uses $\Ot(\frac{1}{\eps} n^{1/k}\log D)$ space at each vertex and $\Ot(\frac{1}{\eps} \log D)$-bit headers, to route a message between any pair of vertices on a $(4k-7+\eps)$-stretch path.
\end{itemize}

%
%
%
%
%
%
%

\end{abstract}

\thispagestyle{empty}

\newpage
\setcounter{page}{1}

\section{Introduction}

\emph{Graph spanners}, \emph{distance oracles} and \emph{compact routing schemes} are fundamental notions in graph theory, data structures and distributed algorithms, respectively, that deal with the natural tradeoff between space and accuracy.

Peleg and Ullman \cite{PeUl89} and Peleg and Sch{\'a}ffer~\cite{PeSc89} introduced the notion of graph spanners.
Let $G=(V,E)$ be an undirected graph. For $u,v\in V$, let $d(u,v)$ be the length of a shortest path between $u$ and $v$. A path between $u$ and $v$ is of $(\alpha,\beta)$-stretch if its length is at most $\alpha\cdot d(u,v)+\beta$. For $(\alpha,0)$-stretch we write $\alpha$-stretch.
A graph $H=(V,E')$, where $E'\subseteq E$, is an $(\alpha,\beta)$-stretch spanner of $G$, if and only if, for every $u,v \in V$ there is an $(\alpha,\beta)$-stretch path between $u$ and $v$ in $H$.
It is known how to efficiently construct a $(2k-1)$-spanner of size $\Ot(n^{1+1/k})$ \cite{ADDJS93, BaSe07}, and this
size-stretch tradeoff is tight assuming the \emph{girth} conjecture of Erd\H{o}s~\cite{Er64} holds.

Thorup and Zwick~\cite{ThZw05} introduced the notion of distance oracles. They showed that it is possible to preprocess a weighted undirected graph in $O(mn^{1/k})$ expected time and create a data
structure of size $O(n^{1+1/k})$ that can answer distance queries with a $(2k-1)$-stretch between any two vertices in $O(k)$ time. For $k=1$ this gives the trivial solution of exact distances with $O(n^2)$ space. For $k=2$, this gives a $3$-stretch distance oracle with $O(n^{3/2})$ space. \patrascu\ and Roditty~\cite{PaRo14} showed that there is another distance oracle between these two solutions. They presented a $(2,1)$-stretch distance oracle with $\Ot(n^{5/3})$ space.
\patrascu, Thorup and Roditty~\cite{PaRoTh12} generalized this and showed that for every $\ell\geq 1$, there is a $(3 - 2 / \ell)$-stretch distance oracle with $\Ot(\ell m^{1+\ell/(2\ell - 1)})$.

Peleg and Upfal~\cite{PeUp89} introduced the notion of compact routing schemes. Awerbuch, Bar-Noy, Linial and Peleg~\cite{DBLP:conf/stoc/AwerbuchBLP89} were the first to distinguish between  \emph{labeled} routing schemes and \emph{name independent} routing schemes. In compact routing schemes there is a preprocessing phase in which a centralized algorithm computes routing tables.
In \emph{labeled} routing schemes a short label is assigned to each  vertex.
After the preprocessing phase ends, messages with an additional short header can be routed between the vertices of the graph in a distributed manner. More specifically, when a message to destination $v$ reaches a vertex $u\neq v$, the routing scheme algorithm executed at $u$ decides locally, based on the routing table of $u$, the message header and the label of $v$ on which link to forward the message. For an integer $k>1$, Thorup and Zwick~\cite{ThZwSPAA01} presented a $(4k-5)$-stretch labeled compact routing scheme  with  routing tables of size $\Ot(n^{1/k})$-bit at each vertex, $o(k \log^2 n)$-bit labels and $o(\log^2 n)$-bit headers.

In a sense finding a good spanner is easier than finding a good distance oracle since a distance oracle has to support efficient distance queries between any pair of vertices. Obtaining a routing scheme is inherently harder than obtaining a distance oracle since decisions  are made locally at a vertex based on a fraction of the whole information, while in distance oracles  the whole information is always available.

Therefore, a general open problem is given a $(\alpha,\beta)$-stretch $S$-space distance oracle can we also obtain a $(\alpha,\beta)$-stretch routing scheme with $O(S/n)$-space routing tables?
Both spanners and distance oracles have the same general tradeoff of $(2k-1)$-stretch and $O(n^{1+1/k})$-space which is optimal under the girth conjecture.
An important challenge is to obtain also a routing scheme with $(2k-1)$-stretch and $\Ot(n^{1/k})$-space routing tables.

Recently, Chechik~\cite{Ch13} made an important breakthrough and improved the stretch of the general routing scheme of Thorup and Zwick~\cite{ThZwSPAA01} from $4k-5$ to $((4 - \alpha)k - \beta)$-stretch,  for some constants $\alpha$ and $\beta$.
The routing tables in her scheme are of size $O(n^{1/k} \log D)$, where $D$ is the normalized diameter of the graph.

Routing schemes, and in particular small stretch routing schemes, have been extensively studied over the last three decades. Cowen~\cite{Cowen01} obtained a $3$-stretch labeled routing scheme with routing tables of $\Ot(n^{2/3})$ space. Eilam \etal~\cite{EiGaPe03} obtained a $5$-stretch labeled routing scheme with routing tables of $\Ot(n^{1/2})$ space. Arias~\etal~\cite{AriasCLRT06} obtained a $5$-stretch name independent routing scheme with routing tables of $\Ot(n^{1/2})$ space.
Thorup and Zwick~\cite{ThZwSPAA01} improved the result of Cowen~\cite{Cowen01} and  obtained a $3$-stretch labeled routing scheme with routing tables of $\Ot(n^{1/2})$ space. Finally, Abraham~\etal~\cite{DBLP:journals/talg/AbrahamGMNT08} obtained a $3$-stretch name independent routing scheme with routing tables of $\Ot(n^{1/2})$ space.

Remarkably, these last two routing schemes are the only optimal routing schemes for general undirected graphs amongst all the possible optimal stretch/space combinations.
Even the new routing scheme of Chechik~\cite{Ch13} is better than the $(4k-5)$-stretch routing scheme only for $k\geq 4$. In particular, for $k=4$ the stretch is approximately $10.52$ instead of $11$ as in the routing scheme of~\cite{ThZwSPAA01}.
Therefore, we are evidence to a  phenomena in which settling on any stretch factor from the range $(3,7)$ cannot help in reducing the routing table size below the $\sqrt n$ barrier.

In this paper we break this long standing barrier and present a $(5+\eps)$-stretch routing scheme that uses $\Ot(\frac{1}{\eps} n^{1/3}\log D)$ space at each vertex and $\Ot(\frac{1}{\eps} \log D)$-bit headers. The label of each vertex is of $O(\log n)$ size. This almost matches the optimal stretch/space combination of $5$-stretch and $O(n^{1/3})$-space routing tables.

Following the $(2,1)$-stretch distance oracle of~\cite{PaRo14} for unweighted graphs, Abraham and Gavoille~\cite{DBLP:conf/wdag/AbrahamG11} presented a $(2,1)$-stretch routing scheme with routing tables of size $\Ot(n^{3/4})$ and in general a
$(4k-6,1)$-stretch $\Ot(n^{\sfrac{3}{(3k-2)}})$-space routing tables. In this paper we almost match the $(2,1)$-stretch distance oracle of~\cite{PaRo14} and present a $(2+\eps,1)$-stretch routing scheme with routing tables of $\Ot(\frac{1}{\eps} n^{2/3})$ size at each vertex and $\Ot(1/\eps)$-bit headers.
In Table~\ref{T-1} we compare our new results for small stretch factors with the previous best bounds.

We obtain our new results by extending a fundamental idea for routing between nearby vertices. In many of the existing routing schemes every vertex $u\in V$ stores the first edge on a shortest path to every vertex in $B(u,\ell)$, the set of $\ell$ closest vertices of $u$. A message is routed from $u$ to $v\in B(u,\ell)$ using the edge that $u$ stored for $v$. It is easy to show that $v\in B(w,\ell)$ for every $w$ that is on a shortest path between $u$ and $v$.
Thus, it is possible to route a message from $u$  to a specific vertex $v\in B(u,\ell)$ on a shortest path by storing at $u$ $O(\log n)$-bits of routing information dedicated for $v$.

At the heart of our new routing techniques is an extension of this fundamental idea. Roughly speaking,  we show that if every $u\in V$ stores the first edge on a shortest path to every $v\in B(u,\ell)$, and an additional routing information on a fixed set of special vertices, it is possible to route a message from $u$ to a specific vertex $v$ that is arbitrarily far away from $u$ on a $(1+\eps)$-stretch path by storing at $u$ $O(\frac{1}{\eps} \log n)$-bits of routing information dedicated for $v$.

Using this extension we present two new routing techniques. Given a partition $\mathcal{U}=\{U_1, \ldots, U_q\}$ of $V$ into $q$ sets each of size $\Ot(n / q)$, our first technique routes a message between any pair of vertices of $U\in \mathcal{U}$ on a $(1 + \eps)$-stretch path with routing tables of size $\Ot(\frac{1}{\eps}(n/q) + q)$. Given a partition $\mathcal{W}=\{W_1, \ldots, W_q\}$ of a set $W\subseteq V$ into $q$ sets each of size $\Ot(|W| / q)$ and assuming that  $U \cap B(u, \Ot( q))\neq \emptyset$, for every $U\in \mathcal{U}$ and $u\in V$,  our second technique routes a message, for every $i\in \{1,\ldots,q\}$, from any source in $U_i$ to any destination in $W_i$ on a $(1 + \eps)$-stretch path with routing tables of size $\Ot(\frac{1}{\eps}(|W|/q) + q)$.

Our $(5+\eps)$-stretch routing scheme is obtained using our second routing technique together with the routing scheme of Thorup and Zwick~\cite{ThZwSPAA01} and a coloring method used by Abraham~\etal~\cite{DBLP:journals/talg/AbrahamGMNT08} and
Abraham and Gavoille~\cite{DBLP:conf/wdag/AbrahamG11}.
Our $(2+\eps,1)$-stretch routing scheme  is obtained using our first routing technique together with ideas of Abraham and Gavoille~\cite{DBLP:conf/wdag/AbrahamG11} and \patrascu\ and Roditty~\cite{PaRo14}.
Using our first technique it is also possible to obtain a name independent routing scheme with stretch $3+\eps$ and routing tables of $\Ot(\sqrt n)$ size.

Our second technique is strong enough to obtain also the following generalizations:
\begin{itemize}
\item For an integer $\ell>1$, a routing scheme for unweighted graphs that uses $\Ot(\ell \frac{1}{\eps} n^{\ell/(2\ell \pm 1)})$ space at each vertex and $\Ot(\frac{1}{\eps})$-bit headers, to route a message between any pair of vertices on a $(3 \pm 2 / \ell + \eps,2)$-stretch path. This almost matches the distance oracles of \patrascu, Thorup and Roditty~\cite{PaRoTh12}.
\item A routing scheme that uses $\Ot(\frac{1}{\eps} n^{1/k}\log D)$ space at each vertex and $\Ot(\frac{1}{\eps} \log D)$-bit headers, to route a message between any pair of vertices on a $(4k-7+\eps)$-stretch path.
\end{itemize}

Our $(4k-7+\eps)$-stretch routing scheme can be used in the $((4 - \alpha)k - \beta)$-stretch routing scheme of Chechik~\cite{Ch13} instead of the $(4k-5)$-stretch routing scheme of Thorup and Zwick in order to obtain slightly better values for $\alpha$ and $\beta$.

\begin{table*}
\centering
    \renewcommand{\arraystretch}{1.2}
    \begin{tabular}{| l | c | c | c |}
    \hline
    Reference & Graph & Stretch & Table Size \\
    \hline \hline
    Abraham and Gavoille~\cite{DBLP:conf/wdag/AbrahamG11}
    & Unweighted
    & $(2, 1)$
    & $\Ot(n^{\sfrac{3}{4}})$\\
    \hline
    Thorup and Zwick~\cite{ThZwSPAA01}, Abraham~\etal~\cite{DBLP:journals/talg/AbrahamGMNT08}
    & Weighted
    & $3$
    & $\Ot(n^{\sfrac{1}{2}})$\\
    \hline
    Thorup and Zwick~\cite{ThZwSPAA01}
    & Weighted
    & $7$
    & $\Ot(n^{\sfrac{1}{3}})$\\
    \hline
    Chechik~\cite{Ch13}
    & Weighted
    & $10.52$
    & $\Ot(n^{\sfrac{1}{4}}\log D)$\\
    \hline\hline
   {\bf Theorem~\ref{T-two-mult-one-add}}
    & Unweighted
    & $(2+\eps, 1)$
    & $\Ot((\sfrac{1}{\eps})\cdot n^{\sfrac{2}{3}})$\\
    \hline
   {\bf Theorem~\ref{T-gen1}}
    & Unweighted
    & $(2\sfrac{1}{3}+\eps, 2)$
    & $\Ot((\sfrac{1}{\eps})\cdot n^{\sfrac{3}{5}})$\\
    \hline
   {\bf Theorem~\ref{T-gen2}}
    & Unweighted
    & $(4+\eps, 2)$
    & $\Ot((\sfrac{1}{\eps})\cdot n^{\sfrac{2}{5}})$\\
    \hline
   {\bf Theorem~\ref{T-5+e}}
    & Weighted
    & $5+\eps$
    & $\Ot((\sfrac{1}{\eps})\cdot n^{\sfrac{1}{3}}\log D)$\\
    \hline
   {\bf Theorem~\ref{T-4k-7+e}}
    & Weighted
    & $9+\eps$
    & $\Ot((\sfrac{1}{\eps})\cdot n^{\sfrac{1}{4}}\log D)$\\
    \hline

    \end{tabular}
    \caption {Previously available routing schemes and our new routing schemes}\label{T-1}
\end{table*}

The rest of this paper is organized as follows. In the next section we present some preliminaries that are needed throughout the paper. In Section~\ref{S-routing-building-blocks} we present our new routing techniques. In Section~\ref{S-new-routing-scheme} we present our new routing schemes for small stretch factors. In Section~\ref{S-new-routing-scheme-generalization} we present our generalized routing schemes.

\section{Preliminaries}

Let $G=(V,E)$ be an $n$-vertices $m$-edges undirected graph. For every $u,v\in V$, let $d(u,v)$ be the length of a shortest path between $u$ and $v$. Let $D = \frac{\max_{u, v}d(u,v)}{\min_{u\neq v}d(u,v)}$ be the \emph{normalized diameter} of $G$.
Let $B(u,\ell)$ be the $\ell$ closest vertices of $u$, breaking ties by lexicographical order of vertex names.
Let $r_{u}(\ell)$ be the largest value for which it holds that every vertex $w\in V$ with $d(u,w)=r_{u}(\ell)$ is in $B(u,\ell)$.
When it will be clear from the context we simply write $r_u$.
Notice that for unweighted graphs it follows from this definition that $d(u,w)\leq r_{u}(\ell) + 1$ for every $w\in B(u, \ell)$.
The following property is a well known property  of such vertex vicinities.

\begin{property}[\cite{DBLP:journals/jal/AwerbuchBLP90},\cite{DBLP:conf/spaa/AbrahamGMNT04}]\label{P-sp}
If $v\in B(u,\ell)$ and $w$ is on a shortest path between $u$ and $v$ then $v\in B(w,\ell)$.
\end{property}

The next Lemma is a direct result of the above property.
\begin{lemma}\label{L-route-to-closest-h}
Let $G=(V, E)$ be a weighted graph and let $\ell > 0$ be an integer. We can construct a routing scheme that uses routing tables of size $O(\ell)$ in each vertex and $O(\log n)$-bit header, such that a message is routed from $u$ to any $v\in B(u, \ell)$ on a shortest path.
\end{lemma}
\begin{proof}
Let $u\in V$. For each $v\in B(u,\ell)$ store at $u$ the first edge on a shortest path to $v$. Assume $v\in B(u,\ell)$ and let $(u,w)$ be the edge saved at $u$ for $v$. When $u$ receives a message to $v$, it forwards  it to $w$. From Property~\ref{P-sp} it follows that $v\in B(w,\ell)$, therefore, $w$ has the next edge on a shortest path to $v$ and the message is routed from $u$ to $v$ on a shortest path.
\end{proof}
Throughout the paper when we say that $u$ saves $B(u,\ell)$ we mean that $u$ can check membership and to retrieve the saved edge in constant time.

For our routing schemes we assume the standard fixed port model as described by Fraigniaud and Gavoille~\cite{FrGa01}.
We use the following tree routing scheme as a building block in our routing schemes:

\begin{lemma}[Tree routing~\cite{ThZwSPAA01, FrGa01}]\label{L-trees}
Given a weighted tree there is a routing scheme that given the label of a destination vertex in the tree, routes from any source vertex in the tree to the destination on the shortest path in the tree. The storage of a node, its label size and the header size are $O(\log^2 n/ \log \log n)$ bits. Each link on the path is obtained in constant time.
\end{lemma}

The notions of bunches and clusters were used by Thorup and Zwick in the context of distance oracles~\cite{ThZw05}.
Let $A\subseteq V$. Let $p_A(u)=\argmin\{ d(u,v)\mid v\in A\}$, breaking ties by lexicographical order of vertex names, and $d(u,A)=d(u, p_A(u))$.

The \emph{bunch} of $v\in V$ is $B_A(v) = \{w\in V \mid d(w,v) < d(v, A)\}$.
The \emph{cluster} of $w\in V$ is $C_A(w) = \{v\in V \mid d(w,v) < d(v, A)\}$. Notice that $w\in B_A(v)$ if and only if $v\in C_A(w)$.

From the definition of clusters it follows that if $u\in C_A(w)$ and $v$ is on a shortest path between $u$ and $w$ then $v\in C_A(w)$. Thus, there is a shortest path tree $T_{C_A}(w)$ rooted at $w$ that spans the vertices of $C_A(w)$. Let $|A| = s$, the size of each bunch is bounded by $O(n / s)$. Thorup and Zwick~\cite{ThZwSPAA01} showed how to simultaneously bound the size of the bunches and clusters.

\begin{lemma}[\cite{ThZwSPAA01}]\label{L-clusters}
Given an integer parameter $s>0$, it is possible to construct a set $A$ with expected size of $2s\log n$, such that $|C_A(w)|\leq 4n/s$, for every $w\in V$.
\end{lemma}
%

We will also make use of the following well known Lemma.
\begin{lemma}[Hitting set~\cite{DBLP:journals/siamcomp/AingworthCIM99, DoHaZw00}]\label{L-hitting-set}
Let $1\leq s \leq n$. Given $k$ sets $S_1,\ldots, S_k$, where $S_i\subseteq V$ and $|S_i|\ge s$, for every $i\in \{1,\ldots,k\}$, we can find a set $H\subseteq V$ of size $\Ot(n/s)$, such that $H\cap S_i \neq \emptyset$.
\end{lemma}


\ignore{
We show that using hitting set of lemma~\ref{L-hitting-set}, we can  to obtain the following routing scheme.
\begin{lemma}
Let $G=(V, E)$ be a unweighted graph and let $h > 0$. We can construct a routing scheme that uses $\max\{h, n/h\}$ information at each vertex and $O(\log n)$ header size, such that a message is routed from a source $u$ to a destination $v$ on a path with a length of $d(u,v) + 2(r_h(u)$.
\end{lemma}

\begin{proof}
We compute a hitting set $H$ (Lemma~\ref{L-hitting-set}) that hits the sets $B(u, h)$, where $u\in V$. For each vertex $w\in H$, we compute shortest path tree $T(u)$ rooted at $w$ and that spans $V$. For each of these trees, we compute tree routing.
\end{proof}
}

Abraham \etal \cite{DBLP:conf/spaa/AbrahamGMNT04} introduced a nice coloring technique and used it to obtain a name-independent routing scheme with stretch $3$ and routing tables of size $\Ot(\sqrt n)$.
Recently, Abraham and Gavoille~\cite{DBLP:conf/wdag/AbrahamG11} used this coloring technique to obtain a labeled routing scheme with stretch $(2,1)$  and routing tables of  size $\Ot(n^{3/4})$.
The next Lemma is implicit in \cite{DBLP:conf/spaa/AbrahamGMNT04,DBLP:conf/wdag/AbrahamG11}.

\begin{lemma}[Coloring~\cite{DBLP:conf/spaa/AbrahamGMNT04,DBLP:conf/wdag/AbrahamG11}]\label{L-coloring}
Let $0 < q \leq n$ be an integer. Let $S_1,\ldots, S_k$ be vertex sets each of size at least $\alpha  q \log n$, where $\alpha$ is a large enough constant. There is a coloring function  $c:V\rightarrow \{1,\ldots,q\}$ that satisfies the following:
\begin{enumerate}
\item For each $i\in \{1,\ldots,k\}$ and $j\in \{1,\ldots,q\}$, there exist $w\in S_i$ such that $c(w) = j$.
\item For each $j\in \{1,\ldots,q\}$ the number of vertices of color $j$ is $O( n / q)$.
\end{enumerate}
\end{lemma}

Abraham \etal \cite{DBLP:conf/spaa/AbrahamGMNT04} showed that such a deterministic coloring function can be computed efficiently.
For the sake of completeness we show that such a coloring function is formed, with high probability, from a random uniform coloring of the vertices of the graph. Assume that the number of sets $k$ is polynomial in $n$. Each vertex is colored by one of the $q$ colors with uniform probability.
Let $S\in \{ S_1,\ldots, S_k\}$ and $j\in \{1,\ldots,q\}$. Consider the event that there is no vertex in $S$ with color $j$. The probability for this event to happen is bounded by $(1-1/q)^s$ which is at most $(1/n)^{\alpha}$. From the union bound we get that the first requirement holds with probability $1-k\cdot q \cdot (1/n)^{\alpha}$.
The expected number of vertices with the same color is $n / q$, and from Chernoff bound it follows that the second requirement holds with high probability. Using the union bound to bound the probability that one of the  requirements does not hold   we get that the above function satisfies, whp, both requirements.

Finally, to simplify the presentation we define  $\tilde x= \alpha x \log n $, where $x>0$ is an integer and  $\alpha$ is  a large enough constant of our choice.

\section{New routing techniques}\label{S-routing-building-blocks}

In this section we present two new techniques for routing between predefined vertex sets. We will use these techniques in the next sections to obtain our new routing schemes. The techniques are presented in two lemmas.
We first show that for a vertex partition that satisfies certain properties it is possible to route efficiently between vertices of the same set in the partition.
%
%


\begin{lemma}\label{L-technique-hitting-set} Let $G=(V, E, w)$ be a weighted undirected graph, where  $w:E\rightarrow \mathds{R}$.
Let $\mathcal{U} = \{U_1, \ldots, U_q\}$ be a partition of $V$ into $q$ sets, each of size $\Ot(n / q)$. For every $\eps >0$, there is a routing scheme that uses a header of $O(\frac{1}{\eps} \log n + \log^2 n /\log \log n)$-bit to route a message between any pair of vertices of the same set in partition $\mathcal{U}$, on a $(1 + \eps)$-stretch path.
The routing table stored at each vertex is of size $\Ot(\frac{1}{\eps}(n/q) + q)$.
\end{lemma}
\begin{proof}
$\;$\\
{\bf Preprocessing:}
For each $u\in V$, we store $B(u, \tilde q)$. We compute a hitting set $H$ as in Lemma~\ref{L-hitting-set} of size $\Ot(n/q)$ that hits for each $u\in V$ the set $B(u, \tilde q)$. For each vertex $w\in H$, we compute a shortest path tree $T(w)$ rooted at $w$ that spans $V$. For each such $T(w)$, we store at each vertex $u\in V$ the tree routing information of $T(w)$ as needed for the tree routing schemes of Lemma~\ref{L-trees}. Since $|H| = \Ot(n/q)$ and the tree routing information is of poly-logarithmic size, the total storage at each vertex for this purpose is $\Ot(n/q)$.

Let $U\in \mathcal{U}$ and let $u\in U$. For every $v\in U$, we compute a sequence $\langle x_1, \ldots, x_{b'} \rangle$ of vertices. The sequence is computed using the following process:

Let $b=\lceil 2/\eps \rceil$. Set $x_0 = u$ and $s=d(u,v) / b$. Let $x_i$ ($i\geq 0$) be the last vertex added to the sequence so far. If $v\in B(x_i, \tilde q)$ then set $x_{i+1}$ to $v$ and stop. If $v\notin B(x_i, \tilde q)$ then let $(y_i,z_i)$ be an edge on a shortest path from $x_i$ to $v$ such that $y_i\in B(x_i, \tilde q)$ and $z_i\notin B(x_i, \tilde q)$. Proceed as follows:

\begin{itemize}
\item If $z_i = v$ then set $x_{i + 1} = y_i$,  $x_{i + 2} = v$ and stop.
\item If $d(x_i, z_i) < s$ then pick $w\in H$ such that $w\in B(x_i, \tilde q)$, set $x_{i+1} = w$ and stop.
\item Otherwise ($d(x_i,z_i) \ge s$) set $x_{i + 1} = y_i$,  $x_{i + 2} = z_i$ and continue to add vertices to the sequence with $x_{i+2}$ being now the last added vertex.
\end{itemize}
%
%
%
%
%

In this process we add vertices to the sequence as long the progress towards $v$ exceeds the threshold value $s$. Therefore, the process has at most $d(u,v)/s \leq b$ rounds. Since in each round at most $2$ vertices are added to the sequence its total size is at most $2b$. At $u$ we store for $v$ the sequence $\langle x_1, \ldots, x_{b'} \rangle$. In case that $x_{b'}\neq v$ then it must be that $x_{b'}\in H$ and we also store at $u$ the label of $v$ in $T(x_{b'})$. Since $|U|=\Ot(n / q)$, the total storage required at $u$ for this information is $\Ot(b\cdot (n / q))$.

{\bf Routing:} Given a set $U\in \mathcal{U}$ and a pair of vertices $u,v\in U$  a message is routed  from $u$ to $v$ as follows.
First, $u$ obtains the sequence $\langle x_1, \ldots, x_{b'} \rangle$ for $v$ from its routing table and adds it to the message header\footnote{In the case that $x_{b'}\not=v$, it adds also the label of $v$ in $T(x_{b'})$ that was stored at $u$ in the preprocessing.}. It also sets $x_1$ to be a temporary target. Let $u'$ be the current vertex on the routing path and let $x_i\neq u'$ be the temporary target. Assume that we are either in the case that $i < b'$ or in the case that $i=b'$ and  $x_{b'}=v$. By the way $\langle x_1, \ldots, x_{b'} \rangle$ was constructed it follows that either  $(u', x_i)\in E$ or $x_i\in B(u', \tilde q)$. In the case that $(u', x_i)\in E$ then we can route from $u'$ to $x_i$ using a direct link\footnote{The standard routing scheme model of \cite{PeUp89} assumes that at $u$ given a neighbor $v$ it is possible to obtain the link that connects $u$ to $v$. We can even avoid this assumption by storing in the sequence edges instead of vertices when needed.}. In the case that $x_i\in B(u', \tilde q)$, it follows from Lemma~\ref{L-route-to-closest-h} that we can route to $x_i$ on a shortest path.
Each time a temporary target $x_i\not= v$ receives the message, it sets $x_{i + 1}$ to be the next temporary target. For the case that $i = b'$ and $x_{b'}\neq v$, it follows from the construction  of $\langle x_1, \ldots, x_{b'} \rangle$  that $x_{b'}\in H$. The routing from $x_{b' - 1}$ to $v$ is done using the tree $T(x_{b'})$.
This routing procedure requires a header of $O(\frac{1}{\eps} \log n + \log^2 n /\log \log n)$-bit.

We now turn to analyze the stretch of the routing path. By the sequence construction, until $x_{b' - 1}$, the message is routed on a shortest path. In the case that $x_{b'}=v$ then the message is routed from $x_{b'-1}$ to $x_{b'}$  on a shortest path as well. In the case that $x_{b'}\neq v$ the message is routed from $x_{b' - 1}$ to $v$ on $T(x_{b'})$, where $x_{b'}\in H\cap B(x_{b'-1}, \tilde q)$. We have added $x_{b'}$ to the sequence because $d(x_{b' - 1}, z_{b' - 1}) < d(u,v) / b $. We also know that $z_{b' - 1}\notin B(x_{b' - 1}, \tilde q)$, thus,  $d(x_{b' - 1}, x_{b'}) \leq d(x_{b' - 1}, z_{b' - 1}) <  d(u,v) / b $. The length of the routing path is bounded by $d(u, x_{b' - 1}) + d(x_{b' - 1}, x_{b'}) + d(x_{b'}, v)$. Using the triangle inequality we bound $d(x_{b'}, v)$ with $d(x_{b' - 1}, x_{b'}) + d(x_{b'-1}, v)$ and get: $$d(u, x_{b' - 1}) + d(x_{b' - 1}, x_{b'}) + d(x_{b'}, v) \leq d(u, x_{b' - 1}) + 2d(x_{b' - 1}, x_{b'}) + d(x_{b' - 1}, v).$$ Since $x_{b' - 1}$ is on a shortest path between $u$ and $v$ it follows that $d(u, x_{b' - 1}) + d(x_{b'-1}, v)=d(u,v)$ and we get:
$$d(u, x_{b' - 1}) + 2d(x_{b' - 1}, x_{b'}) + d(x_{b' - 1}, v) \leq d(u,v) + 2d(u,v) / b.$$ Since $b=\lceil 2/\eps \rceil$ we get that the message traverses a $(1 + \eps)$-stretch path.
\end{proof}

Next we extend Lemma~\ref{L-technique-hitting-set} to the case that the possible destinations are from a subset $W$ of $V$.

\begin{lemma}\label{L-color-routing-W} Let $G=(V, E, \omega)$ be a weighted undirected graph, where  $\omega:E\rightarrow \mathds{R}^{+}$. Let $\mathcal{W} = \{W_1, \ldots, W_q\}$ be a partition of $W\subseteq V$ into $q$ sets, each of size $\Ot(|W| / q)$.
Let $\mathcal{U} = \{U_1, \ldots, U_q\}$ be a partition of $V$ into $q$ sets, such that $U \cap B(u, \tilde q)\neq \emptyset$, for every $U\in \mathcal{U}$ and $u\in V$.
For every $\eps >0$, there is a routing scheme that uses a header size of $O(\frac{1}{\eps}\log (nD))$-bit, and for every $1 \leq i\leq q$ routes a message from any vertex of $U_i$ to any vertex of $W_i$ on a $(1 + \eps)$-stretch path.
The routing table stored at each vertex is of size $\Ot(q + \frac{1}{\eps} \cdot \log (D)(|W|/q))$.
%
%
\end{lemma}
\begin{proof}
$\;$\\
{\bf Preprocessing:}

We first assume that $\omega:E\rightarrow [1,M]$.
Similarly to Lemma~\ref{L-technique-hitting-set}, for each vertex $u\in V$ we store $B(u, \tilde q)$. It requires $\Ot(q)$ space at each vertex.

Let $b = \lceil\frac{2}{\eps}\rceil + 1$. For $1\leq j \leq q$, let $U_j \in \mathcal{U} $ and let $W_j\in \mathcal{W}$. Given a vertex $u\in U_j$ and a vertex $w\in W_j$, we store at $u$ a sequence of vertices that will be used to route towards $w$.
Such a sequence has $O(\log (Mn))$ subsequences. A subsequence $\langle x_1, \ldots, x_{b'} \rangle$ is computed as follows.

Let $x$ be a start vertex and let $s$ be a threshold value that are given as an input to the computation. Set $x_0=x$. Let $x_i$ ($i\geq 0$) be the last vertex added to the subsequence so far. If $w\in B(x_i, \tilde q)$ then set $x_{i+1}$ to $w$ and stop. If $w\notin B(x_i, \tilde q)$ then let $(y_i,z_i)$ be an edge on a shortest path from $x_i$ to $w$ such that $y_i\in B(x_i, \tilde q)$ and $z_i\notin B(x_i, \tilde q)$. Proceed as follows:

\begin{itemize}
\item If $z_i = w$ then set $x_{i + 1} = y_i$,  $x_{i + 2} = w$ and stop.
\item If $d(x_i, z_i) < s$ then there is a vertex $z\in B(x_i, \tilde q)\cap U_j$, set $x_{i+1} = z$ and stop.
\item Otherwise ($d(x_i,z_i) \ge s$) set $x_{i + 1} = y_i$,  $x_{i + 2} = z_i$. If the subsequence is of size $2b$ stop, else continue to add vertices to the subsequence with $x_{i+2}$ being now the last added vertex.
\end{itemize}

Now, the sequence stored at $u$ for $w$ is computed as follows. Let $P(u,w) = \{ u= u_0, u_1, u_2, \ldots, u_t = w \}$ be a shortest path between $u$ and $w$.
We start by adding $u_1$ to the sequence. If $u_1=w$ we stop, otherwise, we add $u_2$ to the sequence as well. Next, if $u_2\neq w$, we start to produce subsequences and add them to the sequence. The first subsequence is produced using $x=u_2$ as the start vertex and $s=2 / b$ as the threshold value. Let $s'$ be the threshold value that was used to produce the last subsequence added so far to the sequence and let $w'$ be its last vertex.
If this subsequence has exactly $2b$ vertices and $w'\neq w$ then a new subsequence is produced  and added to the sequence using $x=w'$ as the start vertex and $s=2s'$ as the threshold value.

Each subsequence added to the sequence has at most $2b$ vertices. In every subsequence, beside maybe the last one, all the vertices are on $P(u,w)$. Moreover, if the $i$th subsequence added to the sequence is not the last subsequence its last vertex is at a distance of at least $2^i$ from its first vertex. Since the length of every shortest path is less than $Mn$ the total number of subsequences cannot exceed $\log (nM)$ and therefore a sequence can have at most $2b \log (Mn)+2$ vertices. Each set in $\mathcal{W}$ is of size $|W| / q$, thus, each $u\in U_j$ stores $|W| / q$ sequences which results in $O(b (\log Mn)|W| / q ) = \Ot(\frac{1}{\eps} (\log M)|W| / q )$ in total.

Next, we show that the preprocessing algorithm works also on graphs with non-negative real edge weights.
Let $G=(V,E,\omega)$ be a weighted graph with $\omega:E\rightarrow \mathds{R}^{+}$ and let  $D$ be its normalized diameter. Let $E'=\{ (u,v)\in E \mid \omega(u,v)=d(u,v)\}$. Let $\omega_{\max} = max \{\omega(u,v) \mid (u,v)\in E'\}$, and let $\omega_{\min} = min \{\omega(u,v) \mid (u,v)\in E'\}$.
For every $(u,v)\in E'$ let $\omega'(u,v)= \omega(u,v) / \omega_{\min}$. It is easy to see that $\omega':E'\rightarrow [1,M]$, where $M=\frac{\omega_{\max}}{\omega_{\min}}$.
Notice that all shortest paths in $G'=(V,E',w')$ are also shortest paths in $G$. Applying the preprocessing algorithm on $G'$ we get that
each $u\in U_j$ stores $|W| / q$ sequences each of size $O(b \log (Mn))$. Since $\omega_{\max}\leq \max_{u,v}d(u,v)$ and $\omega_{\min} = \max_{u\neq v}d(u,v)$ it follows that $M=\frac{\omega_{\max}}{\omega_{\min}}\leq D$.
Therefore, the total space stored at a vertex is $\Ot(\frac{1}{\eps}(\log D) |W| / q )$.
{\bf Routing:} Given a set $U_j \in \mathcal{U}$ and a set $W_j\in \mathcal{W}$, and a pair of vertices $u \in U_j$ and $w \in W_j$, we show how to route a message from $u$ to $w$.
We route towards the last vertex of the sequence stored at $u$ for $w$ in the same manner as in Lemma~\ref{L-technique-hitting-set}. It follows from the sequence construction that all its vertices, beside maybe
the last, are on a shortest path between $u$ and $w$. The last vertex is either $w$ or a vertex that belongs to $U_j$. Let $r_0=u$. For $i\geq 0$, let $r_{i+1}$ be the last vertex of the sequence stored at $r_{i}$ for $w$, and let $r_{i}'$ be the vertex that precedes $r_{i+1}$ in the sequence. If $r_{i+1}\neq w$ then $r_{i+1}\in U_j$. Let $\alpha_i = d(r_i, r_i')$.

\begin{claim}\label{C-1}
For every $i\geq 0$, $d(r_{i+1},w)<d(r_i,w)$. Moreover, $d(r_{i + 1}, w) \leq d(r_i, w) -  (\alpha_i -  \alpha_i / b)$.
\end{claim}
\begin{proof}
Let $i \ge 0$. Consider the routing of a message from $r_i$ to $r_{i+1}$. If $r_{i + 1} = w$ then the message is routed on a shortest path, therefore $d(r_{i+1},w)=d(w,w)<d(r_i,w)$ and also $d(r_{i+1},w) = 0 \leq d(r_i, w) - d(r_i, r_i')\leq d(r_i, w) -  (\alpha_i -  \alpha_i / b)$. Otherwise, we route from $r_i\in U_j$ to a vertex $r_{i + 1}\in U_j$. The sequence stored at $r_i\in U_j$ for $w\in W_j$ contains the first two vertices of a shortest path from $r_i$ to $w$ and then at least one subsequence.
Let $k\geq 1$ be the number of subsequences of this sequence. The last vertex of this sequence is $r_{i+1}$, and $r_i'$ is the vertex that precedes $r_{i+1}$ in the sequence. From the sequence construction it follows that $r_i'$ is on a shortest path between $r_i$ and $w$. Moreover, $r_i'$ is either the last vertex of the $(k-1)$th subsequence or a vertex of the $k$th subsequence or the second vertex of the sequence.

The message is routed from $r_i$ to $r_i'$, and then from $r_i'$ to $r_{i + 1}$. From the way that subsequences are produced it follows that $\alpha_i = d(r_i, r_i') \geq 2 + (2 + 4 + \ldots + 2^{k-1}) = 1 + (1 + 2 + \ldots + 2^{k-1}) = 2^{k}$. Moreover, it follows from the threshold used to produce the $k$th subsequence that $d(r_i', r_{i + 1}) \leq 2^{k} / b \leq \alpha_i / b$.
Since $r_i'$ is on a shortest path between $r_i$ and $w$ it follows that $d(r_i', w)=d(r_i, w) - d(r_i, r_i')$. From the triangle inequality it follows that $d(r_{i + 1}, w) \leq d(r_{i + 1}, r_i') + d(r_i', w)$. We get:

$$d(r_{i + 1}, w) \leq d(r_{i + 1}, r_i') + d(r_i', w)= d(r_{i + 1}, r_i') + d(r_i, w) - d(r_i, r_i') \leq  d(r_i, w) - (\alpha_i -  \alpha_i / b),$$ as required.
%
%

%
%
\end{proof}

Next, we show that the stretch of the routing path is $(1 + \eps)$.
From Claim~\ref{C-1} it follows that the message eventually reaches a vertex $r_t$ that $w$ is the last vertex in its sequence.
For every $0 \leq i < t$, the message is routed from $r_i$ to $r_{i+1}$ on a path of length at most $\alpha_i+\alpha_i / b$. For $i=t$ the message is routed on a shortest path between $r_t$ and $w$.
Let $A=\sum_{i=0}^{t-1}\alpha_i$. The message traverses a path of length at most
$A(1+1/b)+d(r_t,w)$.

From Claim~\ref{C-1}  it follows that $d(r_t,w)\leq d(r_{t-1}, w) - (\alpha_{t-1} -  \alpha_{t-1} / b)$. By applying Claim~\ref{C-1} recursively we get $d(r_t,w)\leq d(u, w) - \sum_{i=0}^{t-1}(\alpha_i-\alpha_i/ b)=d(u, w) - A(1- 1/ b)$. From this we get that $A(1- 1/ b) \leq d(u, w) - d(r_t,w)\leq d(u, w)$ and $A\leq \frac{b}{b-1} d(u, w)$. We can now bound the length of the path that the message traverses:
\begin{align*}
 A(1+1/b)+d(r_t,w)&\leq A(1+1/b)+ d(u, w) - A(1- 1/ b)  \\
 &\leq d(u,w)+2A/b\\
 &\leq d(u,w)+2d(u,w)/(b-1)\\
 &=(1+\frac{2}{b-1})d(u,w)\\
 &\leq (1+\eps)d(u,w)
\end{align*}
\end{proof}

\section{Applications for small stretch routing schemes}\label{S-new-routing-scheme}

In this section we show how to obtain small stretch routing schemes, that almost match the corresponding state-of-the-art distance oracles, using the routing techniques from the previous section. To exemplify the strength of our techniques we first show as a warm-up a new $(3+\eps)$-stretch labeled routing scheme with $\Ot(\frac{1}{\eps})$-bit header size and $\Ot(\frac{1}{\eps}\sqrt n)$-space routing tables. Using the hash function presented in~\cite{DBLP:conf/spaa/AbrahamGMNT04} it is easy to modify this routing scheme to be name-independent. Then we turn to present the main results of this section, our $(2+\eps,1)$-stretch and $(5+\eps)$-stretch routing schemes.

\paragraph{Stretch $3+\eps$.}

Let $q=\sqrt n$. Let $c$ be a coloring function with $q$ colors obtained using Lemma~\ref{L-coloring} with respect to the sets $B(u, \tilde q)$, for every $u\in V$. Let $\mathcal{U} = \{U_1, \ldots, U_q\}$ be the color sets induced by $c$. It follows that $\mathcal{U}$ is a partition and every $U\in \mathcal{U}$ has $\Ot(\sqrt n)$ vertices.
We use Lemma~\ref{L-technique-hitting-set} with $\mathcal{U}$. The information stored for this at every vertex is of size $\Ot(\frac{1}{\eps}\sqrt n)$.
Every $u\in V$ stores $B(u, \tilde q)$. It follows from Lemma~\ref{L-coloring} that $B(u, \tilde q)$ contains a vertex of each color.
For every $i\in \{1,\ldots,q\}$ we store at $u$ a vertex from $B(u, \tilde q)$ of color $i$.
For every $v\in V$ the label of $v$ contains $v$ and $c(v)$. A message is routed from $u$ to $v$ as follows. If $v\in B(u, \tilde q)$ it follows from Lemma~\ref{L-route-to-closest-h} that we can route the message from $u$ to $v$ on a shortest path. If $v \notin B(u, \tilde q)$ then, using the information that is saved at $u$, we route to $w\in B(u, \tilde q)$ with color $c(v)$. Again from Lemma~\ref{L-route-to-closest-h} it follows that we can route from $u$ to $w$ on a shortest path. We route from $w$ to $v$ using Lemma~\ref{L-technique-hitting-set}. We now bound the stretch. From Lemma~\ref{L-technique-hitting-set} it follows that the message is sent from $w$ to $v$ on a path of length at most $(1+\eps)d(w,v)$. The message traverses a path of length at most $d(u,w)+(1+\eps)d(w,v)$. From the triangle inequality it follows that $d(w,v)\leq d(w,u)+d(u,v)$. Also, since $v\notin B(u, \tilde q)$ and $w\in B(u, \tilde q)$ it follows that $d(u,w)\leq d(u,v)$. Therefore, $d(w,v)\leq 2d(u,v)$, and we get that $d(u,w)+(1+\eps)d(w,v)\leq (3+2\eps)d(u,v)$.


\paragraph{Stretch $(2+\eps,1)$.}

We now present a $(2+\eps,1)$-stretch routing scheme with $o(\log^2 n)$-bit labels, $\Ot(1/\eps)$-bit headers and routing tables of $\Ot(\frac{1}{\eps} n^{2/3})$ size. This almost matches the $(2,1)$-stretch distance oracle with $\Ot(n^{5/3})$ total space of~\cite{PaRo14}.

\begin{theorem}\label{T-two-mult-one-add} There is a routing scheme for unweighted undirected graphs with $o(\log^2 n)$-bit labels, that uses $\Ot(\frac{1}{\eps} n^{2/3})$ space at each vertex and $\Ot(1/\eps)$-bit header size, to route a message between any pair of vertices on a $(2 + \eps,1)$-stretch path.
\end{theorem}
\begin{proof}
Let $q=n^{1/3}$. Every $u\in V$ stores $B(u, \tilde q)$. We compute using Lemma~\ref{L-clusters} a set $A\subset V$ of size $\Ot(n^{2/3})$, such that $|C_A(w)| = O(n^{1/3})$ for every $w\in V$. For each $u\in C_A(w)$ we store the routing information for a tree routing scheme of $T_{C_A}(w)$ at $u$. This information is of poly-logarithmic size. Every vertex $u$ is contained in at most $O(n^{1/3})$ cluster trees as $|B_A(u)|=O(n^{1/3})$. Therefore, the storage required at each vertex $u$ for this is $\Ot(n^{1/3})$.

Let $w\in V$. For each $v\in C_A(w)$ we store at $w$ the label of $v$ in the tree routing scheme of $T_{C_A}(w)$. Since $|C_A(w)| = O(n^{1/3})$ the storage required at $w$ for these labels is $\Ot(n^{1/3})$.

For each $w\in A$ let $T(w)$ be a shortest path tree rooted at $w$ that spans $V$. For every $v\in V$ we store routing information of a tree routing scheme of $T(w)$. This allows to route from a vertex $u$ to a vertex $v$ on $T(w)$. Since $|A| = \Ot(n^{2/3})$, each vertex stores $\Ot(n^{2/3})$ information for this purpose.

Let $u\in V$. For every $v\in V$, if $B(u, \tilde q) \cap B_A(v)\neq \emptyset$ we store in a hash table at $u$ in the entry of $v$ a vertex $w \in \argmin_{w'} \{d(u,w') + d(w', v) \mid w'\in B(u, \tilde q)\cap B_A(v) \}$.
There are $\Ot(n^{1 / 3})$ vertices in $B(u, \tilde q)$ each with a cluster of $O(n^{1/3})$ vertices. Let $w\in B(u, \tilde q)$. For every $v\in C_A(w)$ it holds that $B(u, \tilde q)\cap B_A(v)\neq \emptyset$  therefore the storage required for this at every vertex is $\Ot(n^{2/3})$.

Let $c$ be a coloring function with $q$ colors obtained using Lemma~\ref{L-coloring} with respect to the sets $B(u, \tilde q)$, for every $u\in V$. Let $\mathcal{U} = \{U_1, \ldots, U_q\}$ be the color sets induced by $c$. It follows that $\mathcal{U}$ is a partition and every $U\in \mathcal{U}$ has $\Ot(n/q)=\Ot(n^{2/3})$ vertices.
We use Lemma~\ref{L-technique-hitting-set} with the partition $\mathcal{U}$.  The information stored at the routing table of every vertex for this is $\Ot(\frac{1}{\eps}n^{2/3})$.

It follows from Lemma~\ref{L-coloring} that $B(u, \tilde q)$ contains a vertex of each color.
For every $i\in \{1,\ldots,q\}$ we store at $u$ a vertex from $B(u, \tilde q)$ of color $i$ and its distance from $u$.

Finally, the label of $v\in V$ contains the following information: $v$, $c(v)$, $p_A(v)$, $d(v, p_A(v))$ and the label of $v$ in $T(p_A(v))$.

{\bf Routing:} A message is routed from $u$ to $v$ as follows. At $u$ we check in constant time if $B(u, \tilde q)\cap B_A(v) \not= \emptyset$. If this is the case we obtain in constant time the vertex $w \in \argmin_{w'} \{d(u,w') + d(w', v) \mid w'\in B(u, \tilde q)\cap B_A(v) \}$ that was saved at $u$. From Lemma~\ref{L-route-to-closest-h} it follows that we can route the message from $u$ to $w$ on a shortest path. From $w$ to $v$ we can use the tree routing scheme of $T_{C_A}(w)$ as the label of $v$ in the this tree routing scheme is stored at $w$. Thus, the message is routed on a path of length $d(u,w)+d(w,v)$. Next, we show that $d(u,v)=d(u,w)+d(w,v)$. Assume, towards a contradiction, that $w$ is not on a shortest path between $u$ and $v$ and let $P$ be a shortest path between $u$ and $v$. There is a vertex $u'\in B(u, \tilde q)\cap P$ such that $d(u,u')=r_u$, and there is a vertex $v'\in B_A(v)\cap P$ such that $d(v,v')=d(v,p_A(v))-1$ (notice that $B(u, \tilde q)\cap B_A(v) \neq \emptyset$ implies that $v\notin A$ and $d(v,p_A(v))-1 \ge 0$). Since the intersection is not on a vertex from $P$ it must be that $d(u,v)\geq r_u + d(v,v') + 1$. On the other hand $d(u,w)\leq r_u+1$ and $d(v,w)\leq d(v,v')$ and we get that $d(u,w)+d(v,w)\leq d(u,v)$, which contradicts the assumption that $w$ is not on a shortest path between $u$ and $v$.

Consider the case that $B(u, \tilde q)\cap B_A(v) = \emptyset$. We obtain $c(v)$, $p_A(v)$ and $d_A(p_A(v))$ from the label of $v$. A vertex $w\in B(u, \tilde q)$ with $c(w)=c(v)$ is stored at $u$ with $d(u,w)$. If $d(v,p_A(v)) \leq d(u,w)$ then the message is routed from $u$ to $v$ on $T(p_A(v))$. If $d(v, p_A(v)) > d(u, w)$, we route on a shortest path from $u$ to $w$ using Lemma~\ref{L-route-to-closest-h} and then from $w$ to $v$ on a path of length at most $(1+\eps)d(w,v)$ using Lemma~\ref{L-technique-hitting-set}.

If $v\in A$, it follows that $d(v,p_A(v)) = 0 \leq d(u,w)$ and the message is routed on a shortest path to $v$ in $T(v) = T(p_A(v))$. Assume $v\notin A$. As $B(u, \tilde q)\cap B_A(v) = \emptyset$ it follows that $d(u,v) \geq r_u + (d(v, p_A(v)) - 1) + 1 = r_u + d(v, p_A(v))$. If $d(v,p_A(v)) \leq d(u,w)$ then the message is sent on a path of length $d(u,p_A(v))+d(p_A(v),v)$. From the triangle inequality it follows that $d(u,p_A(v)) \leq d(u,v)+d(p_A(v),v)$ and thus $d(u,p_A(v))+d(p_A(v),v)\leq d(u,v)+2d(p_A(v),v)$. Since $d(u,w)\leq r_u+1$ we get that $d(v,p_A(v))-1\leq r_u$. Combining this with the fact that $d(u,v) \geq r_u + d(v, p_A(v))$ we get that $d(u,v)+1 \geq 2d(v, p_A(v))$. Therefore,  $d(u,v)+2d(p_A(v),v)$ is at most $2d(u,v)+1$.
If $d(u,w) < d(v,p_A(v))$ the message is sent on a path of length at most $d(u,w)+(1+\eps)d(w,v)$. From the triangle inequality it follows that $d(w,v) \leq d(u,w)+d(u,v)$ and thus $d(u,w)+(1+\eps)d(w,v)\leq (1+\eps)d(u,v)+(2+\eps)d(u,w)$.
We have $d(u,w)\leq d(v,p_A(v))-1$. We also have $d(u,w)\leq r_u+1$. Combining the last two inequalities with $d(u,v) \geq r_u + d(v, p_A(v))$ we get that $2d(u,w)\leq d(u,v)$. Therefore, the message is routed on a path of length at most  $(2+2\eps)d(u,v)$.
\end{proof}

\paragraph{Stretch $5+\eps$.}

We now present a $(5+\eps)$-stretch routing scheme that uses labels of size $(4\log n)$-bit, $\Ot(1/\eps\log D)$-bit headers and routing tables of $\Ot(\frac{1}{\eps} n^{1/3}\log D)$ size. This almost matches the $5$-stretch distance oracle with $\Ot(n^{4/3})$ total space of~\cite{ThZw05}.
The current best routing scheme with routing tables of $\Ot(n^{1/3})$ size has a stretch of $7$ and $o(\log^2 n)$-bit labels. This routing scheme is obtained by setting $k=3$ in the $(4k-5)$-stretch routing scheme of Thorup and Zwick~\cite{ThZwSPAA01}.

\begin{theorem}\label{T-5+e}  Let $G=(V, E, \omega)$ be a weighted undirected graph, where  $\omega:E\rightarrow \mathds{R}^{+}$. There is a routing scheme for $G$ with $O(\log n)$-bit labels that uses $\Ot(\frac{1}{\eps} n^{1/3}\log D)$ space at each vertex and $\Ot(\frac{1}{\eps} \log D)$-bit headers, to route a message between any pair of vertices on a $(5+\eps)$-stretch path.
\end{theorem}
\begin{proof}
$\;$\\
{\bf Preprocessing:} Let $q=n^{1/3}$. Every $u\in V$ stores $B(u, \tilde q)$. We compute using Lemma~\ref{L-clusters} a set $A\subset V$ of size $\Ot(n^{2/3})$, such that $|C_A(w)| = O(n^{1/3})$ for every $w\in V$.
Let $w\in V$, for each $v\in C_A(w)$ we store at $w$ the label of $v$ in the tree routing scheme of $T_{C_A}(w)$. Since $|C_A(w)| = O(n^{1/3})$ the storage required at $w$ is $\Ot(n^{1/3})$.
For each $u\in C_A(w)$ we store at $u$ the routing information for a tree routing scheme of $T_{C_A}(w)$. This information is of poly-logarithmic size. Notice that $|B_A(u)|=O(n^{1/3})$, thus $u$ is contained in at most $O(n^{1/3})$ trees  and the storage required at $u$ for this is $\Ot(n^{1/3})$.

Let $c$ be a coloring function with $q$ colors obtained using Lemma~\ref{L-coloring} with respect to the sets $B(u, \tilde q)$, for every $u\in V$. Let $\mathcal{U} = \{U_1, \ldots, U_q\}$ be the color sets induced by $c$.
Let $\mathcal{W} = \{W_1, \ldots, W_q\}$ be an arbitrary partition of the set $A$ into $q$ sets each with at most $|A|/q$ vertices. For $w\in A$, let $\alpha(w)$ be the index of the set in partition $\mathcal{W}$ that contains $w$, that is $w\in W_{\alpha(w)}$.
We use Lemma~\ref{L-color-routing-W} with partitions $\mathcal{U}$ and $\mathcal{W}$. The required storage at each vertex is $\Ot(q+\frac{1}{\eps} (\log D) |A|/q)=\Ot(\frac{1}{\eps} (\log D) n^{1/3})$.

It follows from Lemma~\ref{L-coloring} that $B(u, \tilde q)$ contains a vertex of each color.
For every $i\in \{1,\ldots,q\}$ we store at $u$ a vertex from $B(u, \tilde q)$ of color $i$. The storage required for this at $u$ is $\Ot(n^{1/3})$.

Let $v\in V$ and let $(p_A(v), z)$ be the first edge on a shortest path from $p_A(v)$ to $v$. Notice that $v\in C_A(z)$. The label of $v$ contains the following information: $v$, $p_A(v)$, $\alpha(p_A(v))$ and $(p_A(v), z)$.

{\bf Routing:}  A message is routed from $u$ to $v$ as follows. At $u$ we check if $v\in B(u, \tilde q)$. If this is the case then from Lemma~\ref{L-route-to-closest-h} it follows that we can route the message from $u$ to $v$ on a shortest path.
If not we check if $v\in C_A(u)$. If this is the case then the label of $v$ in the tree routing scheme of $T_{C_A}(u)$ is stored at $u$. The message is again routed on a shortest path between $u$ and $v$.

In the case that $v\notin  B(u, \tilde q)$ and $v \notin C_A(u)$ the message is routed on a shortest path to $w\in  B(u, \tilde q)$ that satisfies $c(w)=\alpha(p_A(v))$. The message is routed from $w$ to $p_A(v)$ using Lemma~\ref{L-color-routing-W}. At $p_A(v)$ we obtain $(p_A(v), z)$ from the label of $v$ and forward the message to $z$. At $z$ we have the label of $v$ in the tree routing scheme of $T_{C_A}(z)$ since $v\in C_A(z)$. Using this label and the tree routing scheme of $T_{C_A}(z)$ the message is routed from $z$ to $v$.
The message is routed on a path of length at most $d(u,w)+(1+\eps)d(w,p_A(v))+d(p_A(v),v)$. Since $v\notin  B(u, \tilde q)$ it follows that $d(u,w)\leq d(u,v)$. Since $v \notin C_A(u)$  it follows that $d(v,p_A(v))\leq d(u,v)$.
From the triangle inequality it follows that $d(w,p_A(v))\leq d(u,w)+d(u,v)+d(p_A(v),v)\leq 3d(u,v)$. We get $d(u,w)+(1+\eps)d(w,p_A(v))+d(p_A(v),v)\leq (5+3\eps)d(u,v)$.
\end{proof}

\section{Generalized routing schemes}\label{S-new-routing-scheme-generalization}

\paragraph{Stretch $(3\pm 2/\ell+\eps,2)$.}

We now present a $(3\pm 2/\ell+\eps,2)$-stretch routing scheme that uses $O(\ell \log n)$-bit labels, $\Ot(1/\eps)$-bit headers and routing tables of $\Ot(\ell \frac{1}{\eps} n^{\ell/(2\ell \pm 1)})$ size. This almost matches the distance oracles presented in~\cite{PaRoTh12} for weighted graphs with $(3\pm 2/\ell)$-stretch and $\Ot(\ell m^{1+\ell/(2\ell \pm 1)})$ total space.

We start with a simple Lemma given in~\cite{PaRoTh12} and repeated here for completeness.
\begin{lemma}\label{L-8-PaRoTh12} Let $\ell > 0$ be an integer. Let $\{x_i\}_{i=0}^\ell$ and $\{y_i\}_{i=0}^\ell$ be series of real numbers from $[0, 1]$, such that $x_0 = y_0 = 0$ and $x_i + y_{\ell - i} \leq 1$, for every $i\in \{0,\dots , \ell \}$. There exists $i\in \{0,\dots , \ell -1 \}$ that satisfies $x_i + y_{\ell - i - 1} \leq 1 - 1/\ell$.
\end{lemma}

We now turn to present the $(3- 2/\ell+\eps,2)$-stretch routing scheme.

\begin{theorem}\label{T-gen1} Let $\ell>1$ be an integer. There is a routing scheme for unweighted undirected graphs that uses $O(\ell \log n)$-bit labels, $\Ot(\ell \frac{1}{\eps} n^{\ell/(2\ell - 1)})$ space at each vertex and $\Ot( \frac{1}{\eps})$ header size, to route a message between any pair of vertices on a $(3 - 2 / \ell + \eps,2)$-stretch path.
\end{theorem}
\begin{proof}
$\;$\\
{\bf Preprocessing:}
Let $q = n^{1 / (2\ell - 1)}$ and let $i\in \{0, \ldots, \ell \}$. For every $u\in V$, let $B_i(u) = B(u, \tilde q^i)$.
Every $u\in V$ stores $B_\ell(u)$.

For every $i\in \{0, \ldots, \ell \}$:
\begin{itemize}
\item We compute using Lemma~\ref{L-clusters} a set $L_i\subseteq V$ of size $\Ot(q^{2\ell - i - 1})$, such that $|C_{L_i}(w)| = O(n / q^{2\ell - i - 1})= O(q^i)$, for every $w\in V$.
For each $u\in C_{L_i}(w)$ we store at $u$ the routing information for a tree routing scheme of $T_{C_{L_i}}(w)$. This information is of poly-logarithmic size. Notice that $|B_{L_i}(u)|=O(q^i)$, thus $u$ is contained in at most $O(q^{i})$ trees  and the storage required at $u$ for this is $\Ot(q^{i})$.
\item Let $w\in V$. For each $v\in C_{L_i}(w)$ we store at $w$ the label of $v$ in the tree routing scheme of $T_{C_{L_i}}(w)$. Since $|C_{L_i}(w)| = O(q^i)$ the storage required at $w$ is $\Ot(q^i)$.
\item Let $u\in V$. For every $v\in V$, if $B_i(u)\cap B_{L_{\ell -i}}(v)\neq \emptyset$ we store in a hash table $i$ at $u$ in the entry of $v$ a vertex $w \in \argmin_{w'\in B_i(u)\cap B_{L_{\ell -i}}(v)} \{d(u,w') + d(w', v)  \}$. Since $|C_{L_{\ell -i}}(w)|=O(q^{\ell - i})$ and $|B_i(u)| = \Ot(q^i)$, the storage required for this at every vertex is $\Ot(q^{\ell})$.
\end{itemize}

For every $i\in \{0, \ldots, \ell-1 \}$:
\begin{itemize}
\item Let $c_i$ be a coloring function obtained using Lemma~\ref{L-coloring}  with respect to the sets $B_i(u)$  for every $u\in V$, and let $\mathcal{U}^i = \{U^i_1, \ldots, U^i_{q^i}\}$ be a partition of $V$ into $q^i$ color sets induced by $c_i$. Let $j=\ell-i-1$ and let $\mathcal{W}^j = \{W^j_1, \ldots, W^j_{q^i}\}$ be an arbitrary partition of $L_{j}$ into $q^i$ sets each with at most $|L_{j}|/q^i=q^{2\ell - (\ell - i - 1) -1} / q^i = q^{\ell + i} / q^i = q^\ell $ vertices. For $w\in L_{j}$, let $\alpha_j(w)$ be the index of the set in $\mathcal{W}^j$ that contains $w$, that is $w\in W^j_{\alpha_j(w)}$. We use Lemma~\ref{L-color-routing-W} with the partitions $\mathcal{U}^i$ and $\mathcal{W}^j$. The required storage at each vertex is $\Ot(\frac{1}{\eps}|L_j|/q^i+q^i)=\Ot(\frac{1}{\eps}q^{\ell})$.
\item At $u$ we store for each $j\in \{1, \ldots, q^i \}$ a vertex $w\in B_i(u)$ with $c_i(w)=j$. By Lemma~\ref{L-coloring} such a vertex must exist. The storage required for this at $u$ is $\Ot(q^{i})$.
\end{itemize}

In the label of $v\in V$ we store $v$ and $\{p_{L_i}(v),\alpha_i(p_{L_i}(v)),d(v,p_{L_i}(v)),(p_{L_i}(v),v_i')\}^{\ell-1}_{i=0}$, where $(p_{L_i}(v),v_i')$ is the first edge on a shortest path from $p_{L_i}(v)$ to $v$.

{\bf Routing:} A message is routed from $u$ to $v$ as follows. At $u$ we check  if there exists $i\in \{0, \ldots, \ell \}$ such that  $B_i(u)\cap B_{L_{\ell - i}}(v) \not= \emptyset$. If this is the case we can obtain the vertex $w\in B_i(u)\cap B_{L_{\ell - i}}(v)$ that was saved at $u$. From Lemma~\ref{L-route-to-closest-h} it follows that we can route the message from $u$ to $w$ on a shortest path. From $w$ to $v$ we can use the tree routing scheme of $T_{C_{L_{\ell - i}}}(w)$ as the label of $v$ in the this tree routing scheme is stored at $w$. Thus, the message is routed on a path of length $d(u,w)+d(w,v)$. The vertex $w$ is on a shortest path between $u$ and $v$. The proof for that is identical to the proof given in Theorem~\ref{T-two-mult-one-add} for the case of intersection.

Consider the case that for every $i\in \{0, \ldots, \ell \}$ we have $B_i(u)\cap B_{L_{\ell - i}}(v) = \emptyset$.
Let $a_i=r_u(\tilde q^i)$ and let $b_i=d(v, p_{L_{i}}(v))-1$ if $v\notin L_i$ and $b_i=0$, otherwise. Let $j\in \arg \min_{i\in \{0,\ldots,\ell-1\}} \{a_i +b_{\ell -i -1} \}$. If there is more than one index that achieves the minimum we take the index of the highest value.
Let $k = \ell - j -1$.
We obtain $\alpha_{k}(p_{L_k}(v))$ from the label of $v$ and route on a shortest path to the vertex $w\in B_j(u)$ with $c_{j}(w) = \alpha_{k}(p_{L_k}(v))$, that is stored at $u$. From $w$ the message is routed to $p_{L_k}(v)$
using Lemma~\ref{L-color-routing-W}. From $p_{L_k}(v)$ the message is forwarded to $v_k'$ and then using the tree routing scheme of $T_{C_{L_k}}(v_k')$ it is routed to $v$.
The total length of the path is at most $d(u,w)+(1+\eps)d(w,p_{L_k}(v))+d(p_{L_k}(v),v)$.  Let $d(u,v)=\Delta$.

We first consider the degenerated case of $\Delta=1$. Since there is no intersection we have $B_\ell(u)\cap B_{L_{0}}(v) = \emptyset$, where $B_{L_{0}}(v)=\{v\}$.
Therefore, $a_{\ell-1}\leq a_\ell = 0$ and $ b_0 = 0$, and by the selection rule of $j$ it must be that $j=\ell-1$. The message is routed to $w\in B_{\ell-1}(u)$ and from there to $p_{L_0}(v)=v$. The length of this path is at most $3+\eps$.



Consider now the case that $\Delta > 1$. From the triangle inequality it follows that $d(w,p_{L_k}(v)) \leq d(u,w) + d(u,v) + d(v, p_{L_k}(v))$. If $j = 0$ then $u=w$, $a_0 = 0$, and we need to bound $(2+\eps)(b_{\ell - 1} + 1)+(1+\eps)\Delta$. If $j > 0$ then $d(u,w) \leq a_j + 1$, and we need to bound $(2+\eps)(a_j+b_k+2)+(1+\eps)\Delta$.

For every $i\in \{0, \ldots, \ell \}$ we have $B_i(u)\cap B_{L_{\ell - i}}(v) = \emptyset$ thus $a_i + b_{\ell-i} \leq \Delta - 1$. Let $x_0=a_0 /\Delta$ and $x_i=(a_i + 1)/\Delta$, for $i > 0$. Let $y_i=b_i/\Delta$. Since $x_0=y_0=0$ and $x_i+y_{\ell-i}\leq 1$ we can apply Lemma~\ref{L-8-PaRoTh12} with the series $\{x_i\}_{i=0}^\ell$ and $\{y_i\}_{i=0}^\ell$.

For $j = 0$, we get from Lemma~\ref{L-8-PaRoTh12} that $a_0 + b_{\ell-1} = b_{\ell - 1} \leq (1 - 1 / \ell)\Delta$. Therefore,
\begin{align*}
(2+\eps)(b_{\ell - 1} + 1)+(1+\eps)\Delta &\leq  (2+\eps)((1 - 1 / \ell)\Delta + 1)+(1+\eps)\Delta \\
&=\Delta((1 + \eps) + (1 - 1 / \ell)(2+\eps)) + 2 + \eps\\
&=\Delta(3 + 2\eps - (2 + \eps)/\ell ) + 2 + \eps\\
&\leq \Delta(3 + 3\eps - (2 + \eps)/\ell ) + 2.\\
\end{align*}

If  $j > 0$, it follows from Lemma~\ref{L-8-PaRoTh12} that $a_j + 1 + b_{\ell-j-1}=a_j +  1 + b_{k} \leq (1 - 1 / \ell)\Delta$. Therefore,

\begin{align*}
(2+\eps)(a_j+b_k+2)+(1+\eps)\Delta &\leq (2+\eps)((1 - 1 / \ell)\Delta + 1) + (1+\eps)\Delta\\
&\leq \Delta(3 + 3\eps - (2 + \eps)/\ell ) + 2.\\
\end{align*}

\end{proof}

The routing scheme of $(3 + 2/\ell+\eps,2)$-stretch is very similar to the one presented above. For the sake of completeness we provide it here with all the details.

As before we repeat a Lemma that was given in~\cite{PaRoTh12}.
\begin{lemma}\label{L-8-PaRoTh12-2} Let $\ell > 0$ be an integer. Let $\{x_i\}_{i=0}^\ell$ and $\{y_i\}_{i=0}^\ell$ be series of real numbers from $[0, 1]$, such that $x_0 = y_0 = 0$ and $x_i + y_{\ell - i} \leq 1$, for every $i\in \{0,\dots , \ell \}$. There exists $i\in \{0,\dots , \ell -1 \}$ that satisfies $x_{i + 1} + y_{\ell - i} \leq 1 + 1/\ell$.
\end{lemma}

We now turn to present the  $(3 + 2/\ell+\eps,2)$-stretch routing scheme.

\begin{theorem}\label{T-gen2} Let $\ell > 1$ be an integer. There is a routing scheme for unweighted undirected graphs that uses $O(\ell \log n)$-bit labels, $\Ot(\ell \frac{1}{\eps} n^{\ell/(2\ell + 1)})$ space at each vertex and $\Ot( \frac{1}{\eps})$ headers, to route a message between any pair of vertices on a $(3 + 2 / \ell + \eps,2)$-stretch path.
\end{theorem}
\begin{proof}
$\;$\\
{\bf Preprocessing:}
Let $q = n^{1 / (2\ell + 1)}$ and let $i\in \{0, \ldots, \ell \}$. For every $u\in V$, let $B_i(u) = B(u, \tilde q^i)$.
Every $u\in V$ stores $B_\ell(u)$.
For every $i\in \{0, \ldots, \ell \}$:

\begin{itemize}
\item We compute using Lemma~\ref{L-clusters} a set $L_i\subseteq V$ of size $\Ot(q^{2\ell - i + 1})$, such that $|C_{L_i}(w)| = O(n / q^{2\ell - i + 1})= O(q^i)$, for every $w\in V$.
For each $u\in C_{L_i}(w)$ we store  at $u$ the routing information for a tree routing scheme of $T_{C_{L_i}}(w)$. This information is of poly-logarithmic size. Notice that $|B_{L_i}(u)|=O(q^i)$, thus $u$ is contained in at most $O(q^{i})$ trees  and the storage required at $u$ for this is $\Ot(q^{i})$.
\item Let $w\in V$. For each $v\in C_{L_i}(w)$ we store at $w$ the label of $v$ in the tree routing scheme of $T_{C_{L_i}}(w)$. Since $|C_{L_i}(w)| = O(q^i)$ the storage required at $w$ is $\Ot(q^i)$.
\item Let $u\in V$. For every $v\in V$, if $B_i(u)\cap B_{L_{\ell -i}}(v)\neq \emptyset$ we store in a hash table $i$ at $u$ in the entry of $v$ a vertex $w \in \argmin_{w'\in B_i(u)\cap B_{L_{\ell -i}}(v)} \{d(u,w') + d(w', v)  \}$. Since $|C_{L_{\ell -i}}(w)|=O(q^{\ell - i})$ and $|B_i(u)| = \Ot(q^i)$, the storage required for this at every vertex is $\Ot(q^{\ell})$.
\end{itemize}
For every $i\in \{1, \ldots, \ell \}$:
\begin{itemize}
\item Let $c_i$ be a coloring function obtained using Lemma~\ref{L-coloring}  with respect to the sets $B_{i}(u)$ for every $u\in V$, and let $\mathcal{U}^{i} = \{U^{i}_1, \ldots, U^{i}_{q^{i}}\}$ be a partition of $V$ into $q^i$ color sets induced by $c_i$. Let $j=\ell - i + 1$ and let $\mathcal{W}^j = \{W^j_1, \ldots, W^j_{q^{i}}\}$ be an arbitrary partition of $L_{j}$ into $q^{i}$ sets each with at most $|L_{j}|/q^i=q^{2\ell - (\ell - i + 1) +1} / q^{i} = q^{\ell + i} / q^{i} = q^\ell $ vertices. For $w\in L_{j}$, let $\alpha_j(w)$ be the index of the set in $\mathcal{W}^j$ that contains $w$, that is $w\in W^j_{\alpha_j(w)}$. We use Lemma~\ref{L-color-routing-W} with the partitions $\mathcal{U}^i$ and $\mathcal{W}^j$. The required storage at each vertex is $\Ot(\frac{1}{\eps}|L_j|/q^{i} + q^{i})=\Ot(\frac{1}{\eps}q^{\ell})$.
\item At $u$ we store for each $j\in \{1, \ldots, q^{i} \}$ a vertex $w\in B_{i}(u)$ with $c_{i}(w)=j$. By Lemma~\ref{L-coloring} such a vertex must exist. The storage required for this at $u$ is $\Ot(q^{i})$.
\end{itemize}

In the label of $v\in V$ we store $v$ and $\{p_{L_i}(v),\alpha_i(p_{L_i}(v)),d(v,p_{L_i}(v)),(p_{L_i}(v),v_i')\}^\ell_{i=1}$, where $(p_{L_i}(v),v_i')$ is the first edge on a shortest path from $p_{L_i}(v)$ to $v$.

{\bf Routing:} A message is routed from $u$ to $v$ as follows. As before, at $u$ we check  if there exists $i\in \{0, \ldots, \ell \}$ such that  $B_i(u)\cap B_{L_{\ell - i}}(v) \not= \emptyset$ and if this is the case we do the same as in
Theorem~\ref{T-gen1}.

Consider the case that for every $i\in \{0, \ldots, \ell \}$ we have $B_i(u)\cap B_{L_{\ell - i}}(v) = \emptyset$.
Let $a_i=r_u(\tilde q^i)$ and let $b_i=d(v, p_{L_{i}}(v))-1$ if $v\notin L_i$ and $b_i=0$, otherwise. Let $j\in \arg \min_{i\in \{1,\ldots,\ell\}} \{ a_i +b_{\ell -i +1} \}$.  If there is more than one index that achieves the minimum we take the index of the highest value.
Let $k = \ell - j +1$.
We obtain $\alpha_{k}(p_{k}(v))$ from the label of $v$ and route on a shortest path to the vertex $w\in B_j(u)$ with $c_{j}(w) = \alpha_{k}(p_{L_k}(v))$, that is stored at $u$. From $w$ the message is routed to $p_{L_k}(v)$
using Lemma~\ref{L-color-routing-W}. From $p_{L_k}(v)$ the message is forwarded to $v_k'$ and then using the tree routing scheme of $T_{C_{L_k}}(v_k')$ it is routed to $v$. The total length of the path is at most $d(u,w)+(1+\eps)d(w,p_{L_k}(v))+d(p_{L_k}(v),v)$.  Let $d(u,v)=\Delta$.

We first consider the degenerated case of $\Delta=1$. Since there is no intersection, we have $B_\ell(u)\cap B_{L_{0}}(v) = \emptyset$ and thereby $a_{\ell} = 0$. Also we have that $B_{\ell - 1}(u)\cap B_{L_{1}}(v) = \emptyset$, so it follows that $v\notin C_{L_1}(u)$ and $b_1 = 0$.
We get that $a_{\ell} = b_1 = 0$ and by the selection rule of $j$ it must be that $j=\ell$. The message is routed first to $w\in B_{\ell}(u)$. From there to $p_{L_1}(v)$ and then to $v$. The length of this path is at most $5+\eps$.


Consider now the case that $\Delta > 1$. From the triangle inequality it follows that $d(w,p_{L_k}(v)) \leq d(u,w) + d(u,v) + d(v, p_{L_k}(v))$. Therefore, we need to bound $(2+\eps)(a_j+b_k+2)+(1+\eps)\Delta$.

For every $i\in \{0, \ldots, \ell \}$ we have $B_i(u)\cap B_{L_{\ell - i}}(v) = \emptyset$ thus $a_i + b_{\ell-i} \leq \Delta - 1$. Let $x_0=a_0 /\Delta$ and $x_i=(a_i + 1)/\Delta$, for $i > 0$. Let $y_i=b_i/\Delta$. Since $x_0=y_0=0$ and $x_i+y_{\ell-i}\leq 1$ we can apply Lemma~\ref{L-8-PaRoTh12-2} with the series $\{x_i\}_{i=0}^\ell$ and $\{y_i\}_{i=0}^\ell$ and get $a_j + 1 + b_{\ell-j-1}=a_j +  1 + b_{k} \leq (1 + 1 / \ell)\Delta$. Therefore,
\begin{align*}
(2+\eps)(a_j+b_k+2)+(1+\eps)\Delta &\leq (2+\eps)((1 + 1 / \ell)\Delta + 1) + (1+\eps)\Delta\\
&=\Delta((1 + \eps) + (1 + 1 / \ell)(2+\eps)) + 2 + \eps\\
&=\Delta(3 + 2\eps + (2 + \eps)/\ell ) + 2 + \eps\\
&\leq \Delta(3 + 2/\ell + 4\eps ) + 2.\\
\end{align*}


\end{proof}


\paragraph{Stretch $4k-7+\eps$.}

We present a $(4k-7+\eps)$-stretch routing scheme that uses $O(k\log^2 n/\log \log n)$-bit labels size, $O(\frac{1}{\eps} \log D + \log^2 n /\log \log n)$-bit headers and routing tables of $\Ot(\frac{1}{\eps} n^{1/k}\log D )$ size, where $D$ is the normalized diameter of the graph.
This should be compared with the $(4k-5)$-stretch routing scheme of Thorup and Zwick~\cite{ThZwSPAA01} that has $O(k \log^2 n/\log \log n)$-bit labels, $O(\log^2 n /\log \log n)$-bit headers and routing tables of $\Ot(n^{1/k})$ size.

We start with a short overview of the $(4k-5)$-stretch routing scheme of Thorup and Zwick~\cite{ThZwSPAA01}.
Let $k\geq 1$ and let $A_0\supseteq A_1\supseteq \ldots\supseteq  A_k$ be vertex sets, such that $A_0=V$, $A_k=\emptyset$ and $A_i$ is formed by picking each vertex of $A_{i-1}$ independently with probability $n^{-1/k}$. The expected size of $A_i$ is $n^{1-i/k}$. For every $u\in V$, let $p_i(u)=p_{A_i}(u)$. The \emph{bunch} of $u\in V$ is

$$ B(u)=\cup_{i=0}^{k-1}\{ v\in A_i \setminus A_{i+1} \mid d(u,v) <d(u,p_{i+1}(u)) \},$$ where $d(u,p_k(u))=\infty$.

The size of every bunch is $O(kn^{1/k})$ in expectation or $\Ot(kn^{1/k})$ in the worst case.
The cluster of $w\in A_i \setminus A_{i+1}$ is $C(w)= \{ u\in V \mid d(u,w) < d(u,p_{i+1}(u)) \}$.

In the preprocessing we compute $\{ A_i\}_{i = 0}^{k-1}$. Let $w\in V$ and let $T(w)$ be a shortest path tree rooted at $w$ that spans the vertices of $C(w)$. For every $v\in V$, store $B(v)$ at $v$ and the label of $v$ in the tree routing scheme of $T(w)$ for every $w\in B(v)$. The label of $v$ contains $\{ p_i(v)\}_{i = 0}^{k-1}$ and the tree labels of $v$ in $\{ T(p_i(v))\}_{i = 0}^{k-1}$.

A message is routed from $u$ to $v$ as follows. At $u$ we find the smallest $i$ for which $u\in C(p_i(v))$.  The message is routed using the tree routing scheme of $T(p_i(v))$ from $u$ to $v$ on a path of length at most $d(u,p_{i}(v)) + d(p_{i}(v), v)$. Using the triangle inequality we get $d(u,p_{i}(v)) + d(p_{i}(v), v) \leq 2d(p_{i}(v), v) + d(u,v)$. In~\cite{ThZwSPAA01} they show that for every $j\leq i$ it holds that $d(v, p_j(v)) \leq 2d(u,v)+d(v, p_{j-1}(v))$, and therefore, $d(v, p_j(v)) \leq 2jd(u,v)$. Since $j\leq i\leq k-1$ we get that $2d(p_{i}(v), v) + d(u,v)\leq 2(2(k-1)d(u,v)) + d(u,v) = (4k - 3)d(u,v)$.

In~\cite{ThZwSPAA01} they further reduced the stretch to $4k-5$. Notice that $d(v, p_j(v)) \leq 2(j-1)d(u,v)+d(v, p_{1}(v))$.   Using Lemma~\ref{L-clusters} they compute a set $A_1$ such that $|C(u)| = O(n^{1/k})$ for every $u\in A_0 \setminus A_1$. Each such $u$ stores the set $C(u)$ and for every $v\in |C(u)|$ the label of $v$ in the tree routing scheme of $T(u)$.  To route a message from $u\in A_0 \setminus A_1$ to $v$, $u$ checks if $v\in C(u)$ and if it is then the message is routed to $v$ on $T(u)$. If not, then it follows that $d(v, p_1(v)) \leq d(u, v)$ (as opposed to $d(v, p_1(v)) \leq 2d(u, v)$ before) and the stretch is $4k - 5$. See~\cite{ThZwSPAA01} for the full description.
%
%

We now turn to present our improved routing scheme.

\begin{theorem}\label{T-4k-7+e}  Let $G=(V, E, \omega)$ be a weighted undirected graph, where  $\omega:E\rightarrow \mathds{R}^{+}$. For every $\eps > 0$, there is a routing scheme for $G$ that uses $o(k\log^2 n)$-bit labels, $\Ot(\frac{1}{\eps} n^{1/k}\log D)$ space at each vertex and $\Ot(\frac{1}{\eps} \log D)$-bit header, to route a message between any pair of vertices on a $(4k-7+\eps)$-stretch path.
\end{theorem}
\begin{proof}
$\;$\\
{\bf Preprocessing:} We store the same information as in the $(4k - 5)$-stretch routing scheme of~\cite{ThZwSPAA01}.

Let $q=n^{1/k}$. Additionally, we store $B(u, \tilde q)$  at $u$, for every $u\in V$. This requires $\Ot(n^{1/k})$ space at each vertex.

Let $c$ be a coloring function with $q$ colors obtained using Lemma~\ref{L-coloring} with respect to the sets $B(u, \tilde q)$, for every $u\in V$. Let $\mathcal{U} = \{U_1, \ldots, U_q\}$ be the color sets induced by $c$.
Let $\mathcal{W} = \{W_1, \ldots, W_q\}$ be an arbitrary partition of $A_{k-2}$ into $q$ sets each with at most $|A_{k-2}|/q=\Ot(n^{2/k}/n^{1/k}) = \Ot(n^{1/k})$ vertices. For $w\in A$, let $\alpha(w)$ be the index of the set in partition $\mathcal{W}$ that contains $w$, that is $w\in W_{\alpha(w)}\in \mathcal{W}$.
We use Lemma~\ref{L-color-routing-W} with the partitions $\mathcal{U}$ and $\mathcal{W}$. The required storage at each vertex is $\Ot((\frac{1}{\eps} \log D)|A_{k-2}|/q+q)=\Ot(\frac{1}{\eps} n^{1/k}\log D )$.

It follows from Lemma~\ref{L-coloring} that $B(u, \tilde q)$ contains a vertex of each color.
For every $i\in \{1,\ldots,q\}$ we store at $u$ a vertex from $B(u, \tilde q)$ of color $i$. The storage required for this at $u$ is $\Ot(n^{1/k})$.

The label of every $v\in V$ is composed of the label of $v$ in the $(4k - 5)$-stretch routing scheme and $\alpha(p_{k - 2}(v))$.

{\bf Routing:} A message is routed from $u$ to $v$ as follows. At $u$ we check if $v\in B(u, \tilde q)$. If this is the case the message is routed from $u$ to $v$ on a shortest path. If this is not the case we look for the smallest index $i$ for which $u\in C(p_{i}(v))$. If $i \leq k - 2$, the message is routed to $v$ using the tree routing scheme of $T(p_{i}(v))$ on a path of length at most $d(u,p_i(v))+d(p_i(v),v)$. It follows from the analysis of~\cite{ThZwSPAA01} that  $d(u,p_i(v))+d(p_i(v),v)\leq (4k-9)d(u,v)$, for $i \leq k - 2$.
If $i=k-1$ we route the message to a vertex $w\in  B(u,\tilde q)$ for which $c(w)=\alpha(p_{k - 2}(v))$. The message is routed from $w$ to $p_{k - 2}(v)$ using Lemma~\ref{L-color-routing-W}, and from $p_{k - 2}(v)$ to $v$ using the tree routing scheme of $T(p_{k - 2}(v))$ (recall that the label of $v$ in the tree routing scheme of $T(p_{k - 2}(v))$ is in $v$'s label). The length of the path is at most $d(u, w) + (1 + \eps)d(w, p_{k - 2}(v)) + d(p_{k - 2}(v), v)$. Using the triangle inequality we can bound it with $d(u, v) + (1 + \eps)(2d(u,v) + d(p_{k - 2}(v), v)) + d(p_{k - 2}(v), v)$. From the analysis of~\cite{ThZwSPAA01} it follows that $d(p_{k - 2}(v),v)\leq (2(k-3) + 1) d(u, v)$ and we get that the stretch is $4k-7+\eps$.
\end{proof}

\bibliographystyle{plain}
\bibliography{paper}

\end{document}